\documentclass[conference]{IEEEtran} 
\usepackage{graphicx}
\usepackage[cmex10]{amsmath}
 \usepackage{amsthm} 
\usepackage{amsfonts}
\usepackage{amssymb} 
\usepackage{bbold}
\usepackage{graphicx}%
\usepackage{color}%
\usepackage{algorithmic}
 \usepackage{algorithm}
\newtheorem{mytheorem}{Theorem}

\newtheorem{lem}{Lemma}
\usepackage[english]{babel}
\usepackage{blindtext}
\usepackage{epsfig}
\usepackage{framed}
\usepackage{xcolor}

\title{Enabling the Multi-User Generalized Degrees of Freedom in the Gaussian Cellular Channel}

\author{
\IEEEauthorblockN{Rick Fritschek}
\IEEEauthorblockA{ Lehrstuhl f\"ur Informationstheorie 
    und \\ Theoretische
    Informationstechnik\\
    Technische Universit\"at Berlin, \\
    Einsteinufer 25,
    D--10587 Berlin, Germany\\
    Email: rick.fritschek@tu-berlin.de
}%
\and
\IEEEauthorblockN{Gerhard Wunder}
\IEEEauthorblockA{Fraunhofer
  Heinrich--Hertz--Institut\\
Wireless Communication and Networks \\
Einsteinufer 37, D--10587 Berlin, Germany\\
Email: gerhard.wunder@hhi.fraunhofer.de}

}

\begin{document}

\maketitle
\begin{abstract}
There has been major progress over the last decade in understanding the classical interference channel (IC).
Recent key results show that constant bit gap capacity results can be obtained from linear deterministic models (LDMs).
However, it is widely unrecognized that the time-invariant, frequency-flat cellular channel, which contains the IC as a special case,
possesses some additional generalized degrees of freedom (GDoF) due to multi-user operation. This was proved for the LDM cellular channel
very recently but is an open question for the corresponding Gaussian counterpart. In this paper, we close this gap and provide an achievable sum-rate
for the Gaussian cellular channel which is within a constant bit gap of the LDM sum capacity.
We show that the additional GDoFs from the LDM cellular channel carry over. This is enabled by signal scale alignment.
In particular, the multi-user gain reduces the interference by half in the 2-user per cell case compared to the IC.
\end{abstract}

\section{Introduction}

Starting with the work of Etkin, Tse, and Wang \cite{etkin2008} interest has been growing in constant-gap approximations of network information theory problems. In \cite{etkin2008} the interference channel (IC), a long standing problem which has not yet been solved, was investigated. The capacity region has been found within a constant gap of 1 bit. This astonishing result motivated constant bit-gap capacity investigations of other network topologies and models, e.g. \cite{avestimehr2009}, \cite{Bresler2010}, \cite{Bresler2008}, and \cite{Suvarup2011}. Deriving constant bit-gap results for Gaussian channel models can be a demanding problem. This is partly due to the noise properties of the channel models. A second branch of research therefore investigated approximate Gaussian channel models. This started with the work of Avestimehr, Diggavi, and Tse  \cite{Avestimehr2007}, in which the so-called linear deterministic model (LDM) was first defined and investigated. This network model approximates the Gaussian model by translating the real channel input signals in corresponding binary vectors by means of binary expansion. Power is represented by a shift of the bit vectors, and noise is introduced as a truncation of these vectors. These properties significantly reduce the complexity and difficulty of the model.
Not only does the LDM makes things easier to prove but also hints at correct intuition for its Gaussian counterpart. In \cite{Bresler2008} it was demonstrated that the achievable scheme for the linear deterministic IC (LD IC) is within a constant gap of 42 bit, of the Gaussian channel. Other successful transitions of LDM methods and insights into the Gaussian counterpart include \cite{Bresler2010}, \cite{Sridharan08}, \cite{Suvarup2011}, and \cite{Suvarup2012}. In many cases, the special properties of the LD models provide a possibility of interference alignment on the signal scale. The investigation of the Many-to-One channel \cite{Bresler2010}, was the first to convert this interference alignment to the Gaussian case. Here, lattice coding was used in addition to a layered encoding and successive decoding strategy. The strategy splits the total usable power in many power packages, where each one gets encoded separately by a lattice code. Also each package is viewed as a single Gaussian channel in the decoding process. Since lattice codes obey certain group properties, the addition of many interferers is seen as another codeword at the receiver. This enables interference alignment on the signal scale. The idea was picked up in \cite{Suvarup2011} and \cite{Suvarup2012} which investigated the mixed Many-to-One and One-to-Many models. In these investigations LDM interference alignment strategies were also translated to the Gaussian channel, resulting in new achievable rates. Another interesting model is the interfering MAC (IMAC) and interfering BC (IBC). These terms were defined in \cite{Suh2008}, were it was shown that, using delay properties, a form of multi-user gain can be achieved. Furthermore, with $k-$users, an interference-free transmission is proven by means of degrees of freedom. Another notable investigation is \cite{chaaban2011}, where the capacity region of the IMAC was characterized for several strong interference cases. Also an upper-bound was provided for the whole interference range. A different strategy was used in \cite{Fritschek2014}, where the LD IMAC and LD IBC were investigated. Based on the LD MAC-P2P investigations in \cite{Buhler2011} and \cite{Buhler2012}, multi-user gain was shown and the sum capacity has been found for the weak interference case. 

{\bf Contributions:} In this paper we convert the achievable scheme of \cite{Fritschek2014} to the Gaussian models by modifying the methods of \cite{Bresler2010} and \cite{Suvarup2011}. We show that with layered lattice codes, signal scale alignment can be achieved in the IMAC and IBC. We demonstrate that multi-user gain can be enabled which basically reduces the interference by half (2-user weak interference case) in comparison to the IC. To the best of our knowledge, this is the first demonstration of signal scale alignment and multi-user gain in time-invariant and frequency-flat cellular channel models.

\section{System Model} 
\subsection{IMAC Model}
We consider the Gaussian interfering multiple access channel (IMAC), in which there are two Gaussian MACs interfering among themselves.
There are four independent messages, two in every MAC cell. Each receiver also receives interference from the opposite MAC cell. An illustration of this model is given in \ref{system model}. 
The channel equations are given by 
\begin{equation}
Y^j = \sum_{i=1}^2 \sum_{k=1}^2 h_{ik}^j X_{ik} +Z^j
\end{equation}
where $i$ and $j$ denote the transmitter and receiver cell, respectively and $k$ is the user index.
$Z^j\in \mathcal{CN}(0,1)$ is assumed to be zero mean and unit variance Gaussian noise. Also each transmitted signal has an associated average power constraint $\mathbb{E}\{|X_{ik}|^2\}\leq P$. 

\begin{figure}
\includegraphics[width=0.47 \textwidth]{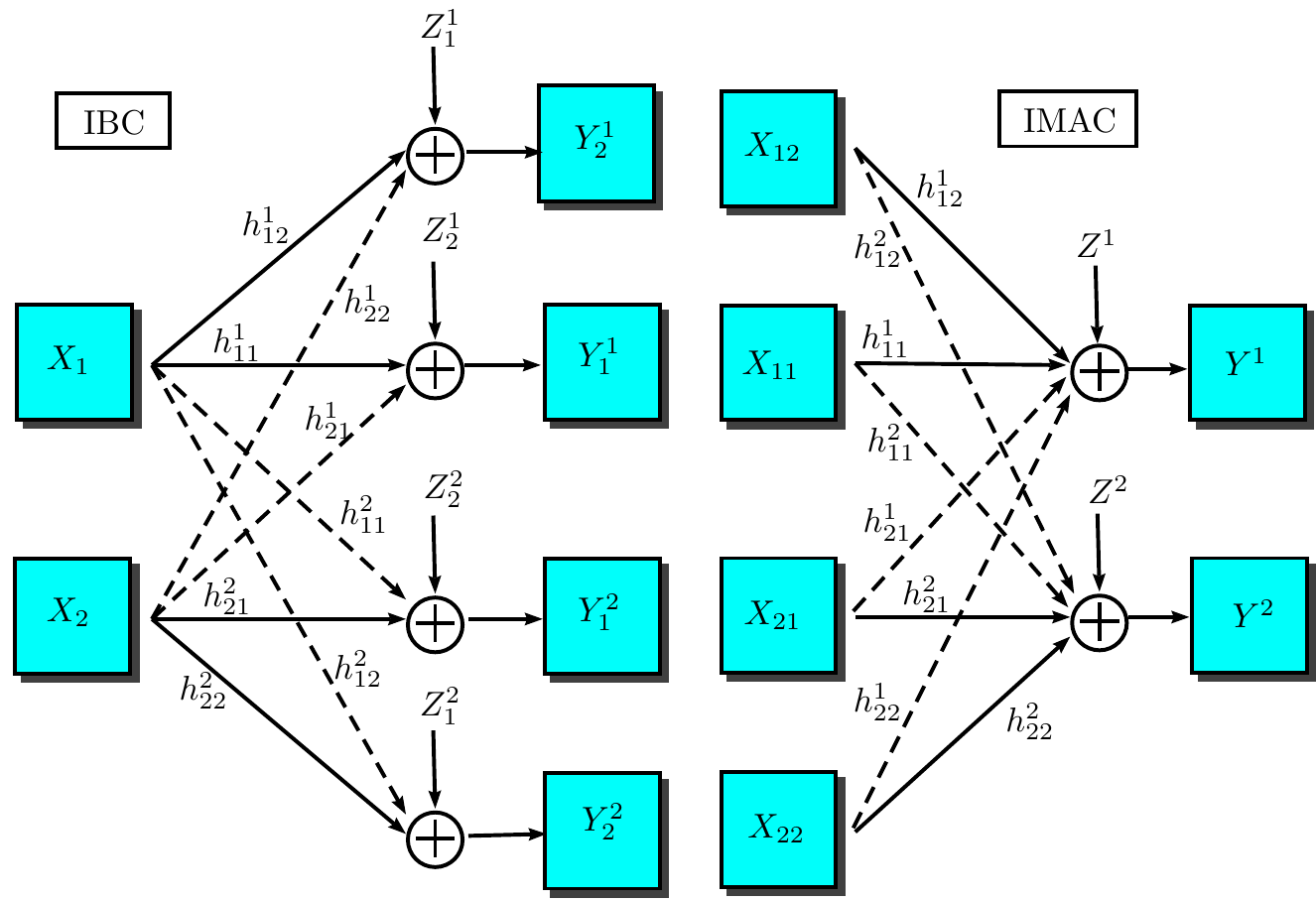}
	\caption{Illustration of the Gaussian IMAC and IBC systems.}
    \label{system model}
\end{figure}

\subsection{IBC Model}
The Gaussian interference broadcast channel (IBC) consists of 2 transmitters and 4 receivers.
As in the IMAC case, there are four independent messages inside the network. Each BC transmitter sends a signal with two messages, one for each receiver of the respective BC-cell. Also every receiver is impaired by interference from the other cell. The IBC channel equations are therefore given by 
\begin{equation}
Y_k^j=\sum_{i=1}^2 h_{ik}^j X_{i} + Z_k^j
\end{equation}
where $Z_k^j\in \mathcal{CN}(0,1)$ is again a zero mean and unit variance Gaussian random variable and $X$ has an associated average power constraint $\mathbb{E}\{|X_{i}|^2\}\leq P$.

{\bf Interference Regime:}
From now on we assume without loss of generality that $h_{i1}^j\geq h_{i2}^j$ for $i=j$.
Also, we assume equal interference strength at the receivers: $h_{i1}^j=h_{i2}^j$ for $i\neq j$. This restriction is justified in the case, when the distance between the two cells is much bigger than the cell dimensions itself.
Furthermore we define two expressions, the signal-to-noise ratio and the interference-to-noise ratio as:
\begin{equation}
|h_{ik}^j|^2P=
\begin{cases}
\mbox{SNR}_{ik} & \text{if }  i=j
\\
\mbox{INR}_{i}^j & \text{if } i\neq j.
\end{cases}
\end{equation}
We also introduce two parameters $\alpha_i,\beta_i$ which combine these ratios with $\mbox{SNR}_{i1}\,=\,P_i$, $\mbox{SNR}_{i2}\,=\,P_i^{\beta_i}$ and $\mbox{INR}_{i}^j\,=\,P_i^{\alpha_j}$. These parameters correspond to $\alpha, \beta$ which are used in the LDM channel model \cite{Fritschek2014}.
Now we can restrict the investigation to the weak interference regime defined through $\alpha_i \in [0,\frac{1}{2}]$ and $\alpha_i < \beta_i$. Note that in the case of $\alpha_i \geq \beta_i$ the IMAC as well as the IBC relapse into the IC in terms of capacity. In both models we assume $P_i>1$. For convenience we use the standard terms: common and private signal, for the part which is seen at both cells and the part which is only received in the intended cell, respectively. Also note that the following techniques work for all channel parameters in the defined regime, except for a singularity at $\beta_i=1$ in which the additional gain for that corresponding cell will be zero.

\section{Linear deterministic approximation}
The LDM models the input symbols at Tx$_i$ as bit vectors $\mathbf{X}_i$. This is achieved by a binary expansion of the real input signal. The resulting bits constitute the new bit vector. The positions within the vector will be referred to as 'levels'. To model the signal impairment induced by noise, the bit vectors will be truncated at noise level and only the n most significant bits are received at $Rx_i$. This is done by shifting the incoming bit vector for $q-n$ positions $\mathbf{Y}=\mathbf{S}^{q-n}\mathbf{X}$, where $\mathbf{S}$ is the shift matrix.
Superposition at the receivers is modelled via binary addition of the incoming bit vectors on the individual levels. 
The channel gain is represented by $n_{ik}^j$-bit levels which corresponds to $\lceil\log |h_{ik}^j|^2 P\rceil$ of the original channel. 
With these definitions the LD IMAC model can be written as
\begin{equation}
\mathbf{Y}^j=\sum_{i=1}^2 \sum_{k=1}^2 \mathbf{S}^{q-n_{ik}^j} \mathbf{X}_{ik}
\end{equation}
and the LD IBC as
\begin{equation}
\mathbf{Y}_k^j=\sum_{i=1}^2 \mathbf{S}^{q-n_{ik}^j} \mathbf{X}_{i}
\end{equation}
where addition is binary.
The sum rate for both models can be upper bounded by 
\begin{equation}
R_\Sigma \leq n_{11}^1+n_{21}^2-\frac{1}{2} n_{21}^1-\frac{1}{2} n_{11}^2
\label{LD_Bound}
\end{equation}
which states that the interference is reduced by half, as proven in \cite{Fritschek2014}.
An achievable sum rate is given with
\begin{equation}
R_\Sigma \leq n_{12}^1+n_{22}^2-n_{21}^1-n_{11}^2 + \phi(n_{21}^1,\Delta_1) +\phi (n_{11}^2,\Delta_2),
\end{equation}
where $\Delta_i=n_{i1}^j-n_{i2}^j$ for $i=j$ is the difference of the direct signals.
The function $\phi$ for $p,q \in \mathbb{N}_0$, following the notation of \cite{Buhler2011}, is defined as
\begin{equation}
\phi(p,q):=
\begin{cases}
q+\frac{l(p,q)q}{2} & \text{if } l(p,q)\text{ is even,}
\\
p-\frac{(l(p,q)-1)q}{2} & \text{if } l(p,q)\text{ is odd},
\end{cases}
\label{phi}
\end{equation}
where $l(p,q):=\lfloor\frac{p}{q}\rfloor\ \mbox{for}\ q>0\ \mbox{and}\ l(p,0)=0$. The optimal construction of the coding matrices is a direct extension of the MAC-P2P scheme in \cite{Buhler2011}. 
The maximum sum rate is achieved by orthogonal coding. This means that there is either an independent bit or nothing send from a specific level. At the receiver there is no overlap with levels used by another transmitter. This technique can be interpreted as interference alignment, because the levels for the LD IMAC, are chosen in the following way. Both interference signals of a cell align at the receiver as much as possible in the unused levels of the direct transmission. In the LD IBC, levels of the common part are chosen such that they align at the receivers of the unintended cell as much as possible on levels, which receive the private message intended for the opposite receiver. See figure \ref{MAC_MAC_SCHEMA} for an example.
\begin{figure}
\centering
\includegraphics[scale=0.5]{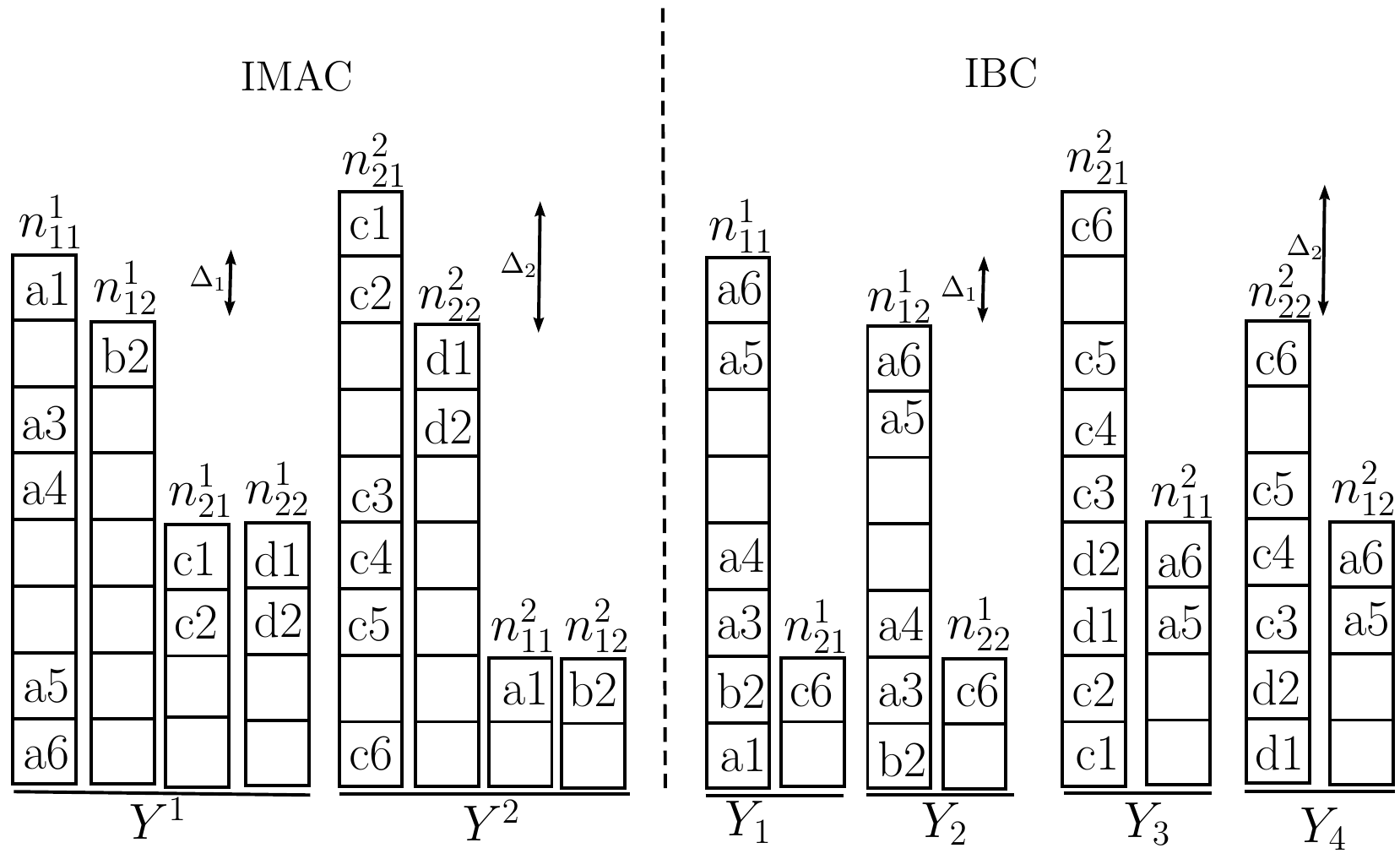}
\caption{{\bf LD IMAC:} An example for a LD IMAC scheme which achieves the upper bound is presented in the figure. The MAC-cell has $n_{11}^1=8$ bit levels and $n_{12}^1=7$ bit levels and generates interference, at the other side, of $n_{11}^2=n_{12}^2=2$ bit levels. Whereas the other MAC cell has $n_{21}^2=9 $ and $n_{22}^2= 7$ bit levels and generates $n_{21}^1=n_{22}^1=4$ bit levels interference.
{\bf LD IBC:} A coding scheme of the LD IBC system model is shown. The example is chosen such that it depicts the exact dual case to the LD IMAC model as in the previous example. The basic strategy to get a coding scheme for the LD IBC case out of the LD IMAC case is shown. The MAC components got merged and the coding vectors inverted.
{\bf Sum rate:}
Both schemes above yield a sum rate of 14 bit levels and therefore reach the upper bound which can be calculated with $R_{\Sigma}\leq n_{11}^1+n_{21}^2 -\frac{n_{21}^1}{2}- \frac{n_{11}^2}{2}$.
}
\vspace{-0.3cm}
\label{MAC_MAC_SCHEMA}
\end{figure}

\section{Coding Schemes for IBC and IMAC}
\subsection{Example of the symmetric restricted IMAC}
{\bf Power partitioning:\ }
In this example we consider the fairly restricted symmetric IMAC channel, where $\alpha=0.5$ and $\beta=0.75$ and both cells have a scaled direct channel gain of 1.
The power as observed at each receiver is partitioned into $P^{(1-\beta)}$ intervals. These intervals play the role of bit levels in the LDM. 
Signal power $\theta_l$ is defined as 
\begin{IEEEeqnarray}{rCl}
\theta_l & = & q_{l-1}-q_l\IEEEnonumber\\
&= & P_i^{1-(l-1)(1-\beta_i)}-P_i^{1-l(1-\beta_i)}
\label{theta_k_fractions}
\end{IEEEeqnarray}
with $l$ indicating the specific level (Fig. \ref{Powersplit}). Each user $k$ of cell $i$ decomposes its signal into a sum of independent sub-signals $X_{ik}=\sum\limits_{l=1}^{l_{\max}} X_{ik}(l)$. For every sub-signal a lattice code is chosen as described in \cite{Loeliger}, such that the spherical shaping region of the lattice has an average power per dimension of $\theta_l$ and is good for channel coding. Moreover, aligning sub-signals use the same code (with independent shifts).

\begin{figure}
\centering
\includegraphics[scale=.7]{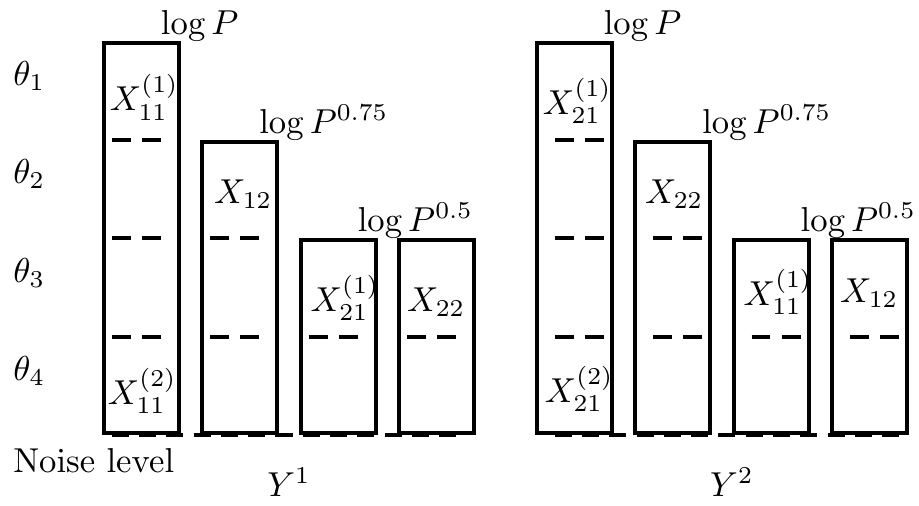}
\caption{Illustration of power partitioning, with the resulting 4 signal power levels and level use for coding}
\label{Powersplit}
\vspace{-1em}
\end{figure}

In \cite{Loeliger} it was shown that under the restriction of
\begin{equation}
R\leq \log \left(\frac{P}{N}\right),
\label{Decodingbound}
\end{equation}
a lattice code $(\gamma\Lambda_C+v) \cap S$ exists with arbitrary small error probability. This code consists of a lattice $\Lambda_C \in \mathbb{R}^n$, a scaling factor $\gamma \in \mathbb{R}$, a translation $v\in \mathbb{R}^n$ and a spherical shaping region $S\subset\mathbb{R}^n$ with power $P$ per dimension. $N$ denotes the noise variance per dimension. 

{\bf Decoding procedure:\ }
Decoding occurs per level, treating subsequent levels as noise. If a sub-signal was decoded it gets subtracted from the remaining signal and the process continues with the next level. In case of an interference-affected level, only the sum of both sub-signals gets decoded and subtracted. Because each level is treated as a Gaussian point-to-point channel, decodability is assured providing that the rate is chosen appropriately according to (\ref{Decodingbound}). With a signal power of $\theta_l$, it only remains to specify the total noise of each level, consisting of the Gaussian noise at the receiver and the signal power of all subsequent levels, including the interference. An achievable rate for each level in the example (Fig. \ref{Powersplit}) is therefore
\begin{IEEEeqnarray*}{rCl}
r_1&= &\log \left(\frac{P-P^{0.75}}{1+4P^{0.75}} \right)^+,\ r_2= \log \left(\frac{P^{0.75}-P^{0.5}}{1+4P^{0.5}} \right)^+\\
r_3&= &\log \left(\frac{P^{0.5}-P^{0.25}}{1+4P^{0.25}} \right)^+,\ r_4= \log \left(\frac{P^{0.25}-1}{5} \right)^+.
\end{IEEEeqnarray*}
The total achievable rate is the summation over all levels
\begin{equation}
R_\Sigma = 2\min(r_1,r_3)+2\min(r_2,r_3)+2r_4,
\end{equation}
where the minima are necessary, because $X_{11}(1),X_{12}$ and $X_{21}(1),X_{22}$ need to be decodable at both receivers. The total achievable rate is therefore 
\begin{IEEEeqnarray*}{rCl}
R_\Sigma &=& 4r_3+2r_4\\
&=&4\log \left(\frac{P^{0.5}-P^{0.25}}{1+4P^{0.25}} \right)^++2\log \left(\frac{P^{0.25}-1}{5} \right)\\
&>&4\log \left(\frac{P^{0.5}-P^{0.25}}{1+P^{0.25}} \right)-4 \log 4 +2\log \left(\frac{P^{0.25}-1}{5} \right)\\
&>&4\log \left(1+\frac{P^{0.5}-P^{0.25}}{1+P^{0.25}} \right)+2\log \left(\frac{P^{0.25}-1}{5} \right)\\
&&-\:4(1+ \log 4)\\
&=&4\log \left(\frac{P^{0.5}+1}{1+P^{0.25}} \right)+2\log \left(\frac{P^{0.25}-1}{5} \right)-4(1+ \log 4)\\
&>&4\log \left(\frac{P^{0.5}}{2P^{0.25}} \right) +2\log \left(P^{0.25} \right)-2(7+\log 5)\\
&=& 4\log P^{0.5}-2 \log P^{0.25} -2(7+\log 5)\\
&=& 2\log P -\log P^{0.5}-2(7+\log 5)
\end{IEEEeqnarray*}
where the result of lemma \ref{lemma_minterm} is used. With the definition of $\alpha=0.5$ and $n_{ik}^j=\lceil\log |h_{ik}^j|^2 P\rceil$ one can see, that the proposed scheme can achieve the upper bound in (\ref{LD_Bound}) within a constant gap of $2(7+\log 5)$ Bits. 

\begin{lem}
The decoding bound for the direct path rate $r_d$ is always greater than the decoding bound for the interference path rate $r_i$
\begin{equation}
r_{\text{d}}\geq r_{\text{i}}.
\end{equation}
\label{lemma_minterm}
\end{lem}
\vspace{-0.3cm}
\begin{proof}
{\bf IMAC:}
A decoding bound for the direct rate is given as $r_d =\log \left( \frac{\theta_l}{1+4q_l}\right)$.\\
The observed interference power is therefore $P^{\alpha-1}\theta_l$ with a corresponding noise of $1+4q_lP^{\alpha-1}$.
Hence, the decoding bound for the interference rate is
\begin{IEEEeqnarray*}{rCl}
r_i &=&\log \left( \frac{P^{\alpha-1}\theta_l}{1+4q_lP^{\alpha-1}}\right)=\log \left( \frac{\theta_l}{P^{1-\alpha}+4q_l}\right).
\end{IEEEeqnarray*}

Subtracting both rates one gets 
\begin{IEEEeqnarray*}{rCl}
r_d-r_i &=&\log \left( \frac{\theta_l}{1+4q_l}\right)-\log \left( \frac{\theta_l}{P^{1-\alpha}+4q_l}\right)\\
		& = & \log \left( \frac{P^{1-\alpha}+4q_l}{1+4q_l}\right) \geq 0,
\end{IEEEeqnarray*}
since $\alpha \leq 1$, $P\geq 1$ and therefore $P^{1-\alpha} \geq 1$.
The IBC case follows on the same lines.
\end{proof}

\subsection{IMAC / IBC - The general (weak) interference case}
We proceed in parallel to section IV with the general case.
In the general case, $\beta_i$ and $\alpha_i$ can be any value in the defined regime and therefore any number of levels can be needed.
The power splitting is done as in the example with $P^{(1-\beta_i)}$ intervals and signal power is given by (\ref{theta_k_fractions}).
The choice of codeword decomposition and level usage is dependent on the underlying LDM scheme.
As in section IV, the sub-signal codewords can be decoded providing a rate of (\ref{Decodingbound}).
It remains to specify the effective noise per level for the IMAC and IBC case.
The effective noise becomes 
\begin{equation}
N_i(l)=1+uq_l=1+uP_i^{1-l(1-\beta_i)},
\label{noise_term}
\end{equation}
where $u=2$ and 4 for the IBC and IMAC, respectively. Henceforth, $i\neq j$ for $i,j\in \{1,2\}$ in all equations.
For the case of $\alpha < 0.5$, there is additional available power in the private part above $P_i^{\alpha_j}$. The additional rate can be expressed with
\begin{equation}
R_{A_i}=
\begin{cases}
R_{f_i} &\text{ if } \alpha < 0.5
\\
0 &  \text{ if } \alpha = 0.5,
\end{cases}
\end{equation}
where
\begin{IEEEeqnarray*}{rCl}
R_{f_i}&=& \log \left(\frac{P_i^{1-\lfloor L_j \rfloor(1-\beta_i)}-P_j^{\alpha_i}}{1+uP_j^{\alpha_i}} \right)^+\\
&>&\log \left(\frac{P_i^{1-\lfloor L_j \rfloor(1-\beta_i)}}{2P_j^{\alpha_i}} \right) - (1+\log u)\\
&= &\log P_i^{1-\lfloor L_j \rfloor(1-\beta_i)}-\log P_j^{\alpha_i} -(2+\log u).
\end{IEEEeqnarray*}
Furthermore, if $L_i=\frac{\alpha_i}{1-\beta_j}\notin \mathbb{N}$, the alignment structure has a remainder term allocated at the lowest power level. The additional rate for the part is dependent on the alignment structure and can be written as
\begin{equation}
R_{\text{R}_i}:=
\begin{cases}
P_i^{\alpha_j-\lfloor L_j \rfloor(1-\beta_i)}-1 & \text{if } \lfloor L_j \rfloor\text{ is odd,}
\\
0 & \text{if } \lfloor L_j \rfloor\text{ is even},
\end{cases}
\end{equation}
because in the even case the additional signal levels cannot be used due to the specific alignment scheme.
The rates for each level inside the alignment structure can be expressed as

\begin{IEEEeqnarray*}{rCl}
R_{\text{I}_i}(l) &=& \log \left( \frac{\theta_l}{1+uq_l}\right)^+\\
&=& \log \left( \frac{P_i^{1-(l-1)(1-\beta_i)}-P_i^{1-l(1-\beta_i)}}{1+uP_i^{1-l(1-\beta_i)}}\right)^+.
\end{IEEEeqnarray*}

\begin{mytheorem}
The overall achievable sum rate $R_\Sigma$ is 
\begin{IEEEeqnarray*}{rCl}
R_{\Sigma}&>&\log P_1^{\beta_1} -\log P_1^{\alpha_2}+\log P_2^{\beta_2} -\log P_2^{\alpha_1}\\
&&+\:\phi(\log P_1^{\alpha_2},\log P_1^{(1-\beta_1)})+\phi(\log P_2^{\alpha_1},\log P_2^{(1-\beta_2)})\\
&&-\:9-6\lfloor L_2 \rfloor-6\lfloor L_1 \rfloor
\end{IEEEeqnarray*}

with $\phi$ defined as in (\ref{phi}). Note that l(p,q) is equivalent to $\lfloor L_i \rfloor$.

\label{mytheorem1}
\end{mytheorem}
\begin{proof}
The total achievable sum rate is the summation over the interference alignment structure $R_{\text{I}_i}$, and additional parts $R_{\text{A}_i}$ and $R_{\text{R}_i}$. 
We show the proof exemplary for the IMAC channel in case that $\alpha<0.5$ and $\lfloor L \rfloor=even$. The IBC proof follows on the same lines. Also note that the rates for the alignment structure are solely of the parts below $P_i^{\alpha_j}$. This is necessary for ensuring decodability of both, direct and interference sub-signals (see lemma 1).
\begin{IEEEeqnarray*}{rCl}
R_{\Sigma}&=& \sum\limits_i^{2}  R_{\text{I}_i}+R_{\text{A}_i}+R_{\text{R}_i}\\
&=& \sum\limits_l^{\lfloor L_2 \rfloor} \mathbb{1}_{l: \scriptscriptstyle\rm odd\  }2 \log \left(\frac{P_1^{\alpha_2-1}\theta_l }{1+4q_lP_1^{\alpha_2-1}} \right)^+\\
&&+\:\mathbb{1}_{l: \scriptscriptstyle\rm even\ } \log_2 \left(\frac{P_1^{\alpha_2-1}\theta_l }{1+4q_lP_1^{\alpha_2-1}} \right)^+\\
&&+\:\sum\limits_l^{\lfloor L_1 \rfloor} \mathbb{1}_{l: \scriptscriptstyle\rm odd\  }2 \log \left(\frac{P_2^{\alpha_1-1}\theta_l }{1+4q_lP_2^{\alpha_1-1}} \right)^+\\
&&+\:\mathbb{1}_{l: \scriptscriptstyle\rm even\ } \log_2 \left(\frac{P_2^{\alpha_1-1}\theta_l }{1+4q_lP_2^{\alpha_1-1}} \right)^+ +\sum\limits_i^{2} R_{A_i}+R_{R_i}\\
&\overset{(a)}{>}&\log P_1^{\lfloor L_2 \rfloor(1-\beta_1)} +\frac{\lfloor L_2 \rfloor}{2}\log P_1^{(1-\beta_1)}+\log P_2^{\lfloor L_1 \rfloor(1-\beta_2)} \\
&&+\: \frac{\lfloor L_1 \rfloor}{2}\log P_2^{(1-\beta_2)}+\sum\limits_i^{2} R_{A_i}+R_{R_i}\\
&>&\log P_1^{\lfloor L_2 \rfloor(1-\beta_1)} +\frac{\lfloor L_2 \rfloor}{2}\log P_1^{(1-\beta_1)}+\log P_2^{1-\lfloor L_1 \rfloor(1-\beta_2)} \\
&&+\:\frac{\lfloor L_1 \rfloor}{2}\log P_2^{(1-\beta_2)}+\log P_1^{1-\lfloor L_2 \rfloor(1-\beta_1)}-\log P_1^{\alpha_2} \\
&&-\:\log P_2^{\alpha_1}+\log P_2^{\lfloor L_1 \rfloor(1-\beta_2)}-8+\sum\limits_i^2 R_{R_i}-6\lfloor L_i \rfloor\\
&=&\log P_1^{\beta_1} -\log P_1^{\alpha_2}+\log P_2^{\beta_2} -\log P_2^{\alpha_1}\\
&&+\:\phi(\log P_1^{\alpha_2},\log P_1^{(1-\beta_1)})+\phi(\log P_2^{\alpha_1},\log P_2^{(1-\beta_2)})\\
&&-\:8-6\lfloor L_2 \rfloor-6\lfloor L_1 \rfloor
\end{IEEEeqnarray*}

where (a) follows with 
\begin{IEEEeqnarray*}{rCl}
R_{\text{I}_i}&= & \sum\limits_l^{\lfloor L_j \rfloor} \mathbb{1}_{l: \scriptscriptstyle\rm odd\ }2 \log \left(\frac{P_i^{\alpha_j-(l-1)(1-\beta_i)}-P_i^{\alpha_j-l(1-\beta_i)}}{1+4P_i^{\alpha_j-l(1-\beta_i)}} \right)^+\\
&&+\:\mathbb{1}_{l: \scriptscriptstyle\rm even\ } \log_2 \left(\frac{P_i^{\alpha_j-(l-1)(1-\beta_i)}-P_1^{\alpha_j-l(1-\beta_i)}}{1+4P_i^{\alpha_j-l(1-\beta_i)}} \right)^+\\
&>&-4.5\lfloor L_j \rfloor+\sum\limits_k^{\lfloor L_j \rfloor}  \log \left(\frac{P_i^{\alpha_j-(k-1)(1-\beta_i)}}{2P_i^{\alpha_j-k(1-\beta_i)}} \right) \\
&&+\:\mathbb{1}_{l: \scriptscriptstyle\rm odd\ } \log \left(\frac{P_i^{\alpha_j-(l-1)(1-\beta_i)}}{2P_i^{\alpha_j-l(1-\beta_i)}} \right)\\
&=&-5.5\lfloor L_j \rfloor+\log P_i^{\lfloor L_j \rfloor(1-\beta_i)} \\
&&+\:\sum\limits_l^{\lfloor L_j \rfloor}\mathbb{1}_{l: \scriptscriptstyle\rm odd\ } \log \left(\frac{P_i^{\alpha_j-(l-1)(1-\beta_i)}}{2P_i^{\alpha_j-l(1-\beta_i)}} \right)\\
&\overset{(*)}{>}&\log P_i^{\lfloor L_j \rfloor(1-\beta_i)} +\frac{\lfloor L_j \rfloor}{2}\log P_i^{(1-\beta_i)}-6\lfloor L_j \rfloor
\end{IEEEeqnarray*}
and (*) is valid for the case $\lfloor L_j \rfloor=$ even. The odd case is $R_{\text{I}_i}>\log P_i^{\lfloor L_j \rfloor(1-\beta_i)} +\frac{\lfloor L_j \rfloor+1}{2}\log P_i^{(1-\beta_i)}-6\lfloor L_j \rfloor+1$.
\end{proof}

One can see that the achievable rate of the Gaussian channel is within a constant-gap of $2(2+\log 4)+6\lfloor L_1 \rfloor+6\lfloor L_2 \rfloor$ of the LDM rate using the correspondence $n_{ik}^j=\lceil\log |h_{ik}^j|^2 P\rceil$. See figure \ref{system_model} for a graphical comparison.

\begin{figure}
\includegraphics[width=0.45 \textwidth]{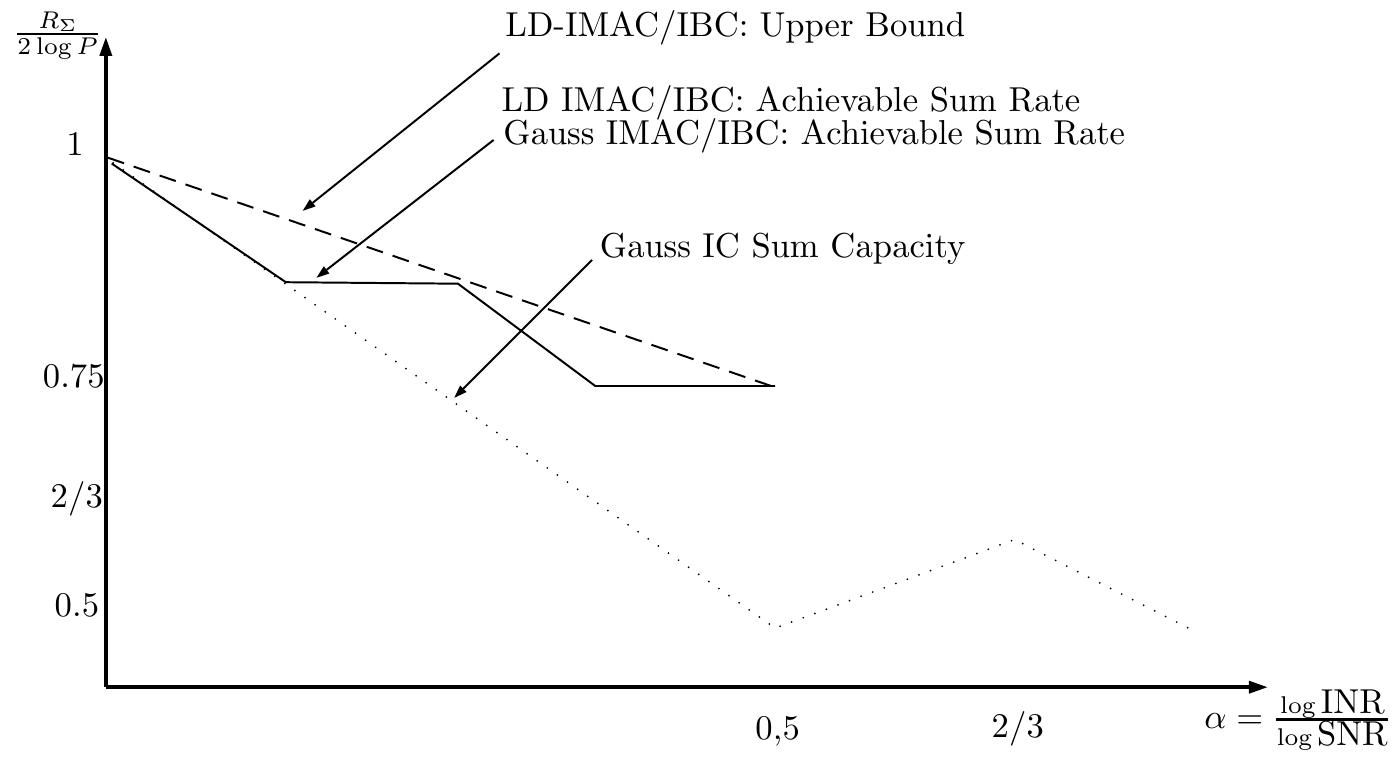}
	\caption{Illustration of the achievable rate in comparison to the LD case. Direct signals and interference is assumed to be symmetrical for clear presentation and $\beta=7/8$.
	At $\lfloor L \rfloor=L=\text{even}$, the achievable scheme reaches the LDM upper bound within a constant gap. }
    \label{system_model}
\end{figure}

\section{Conclusions}
We have demonstrated a technique for signal scale alignment in cellular multi-user networks. We provide the achievable rate for the IMAC and IBC. The rate gap is dependent on the number of power-levels and therefore on the strength-difference between the direct signals. This is consistent with the intuition, since for the limit case of $P_i=P_i^{\beta_i}$, the cellular multi-user model relapses into the IC. In that case there is no power difference to exploit for signal scale alignment. Furthermore, we have demonstrated that multi-user gain is achievable in time and frequency constant models. Multi-user gain reduces the interference by half in the 2-user case in comparison to the IC. 
Moreover, straight forward extensions of the techniques in \cite{Bresler2008} show, that the sum rate of the Gaussian IMAC and IBC channel can be upper bounded by the LDM channel bounds within a constant gap.
We believe that this investigation provides crucial insights for cellular networks in general and can be helpful in future investigations. 



\bibliographystyle{./IEEEtran}
\bibliography{./ref}

\end{document}